\documentclass[letterpaper, 10 pt, conference]{ieeeconf} 
\IEEEoverridecommandlockouts                              
\usepackage{amsmath}
\usepackage{amssymb}
\usepackage{graphicx}
\usepackage{cite}
\usepackage{subfigure}

\title{\bf
Verifying Compositional Refinement of Assume/Guarantee Contracts using Linear Programming}

\newcommand{\C}{\mathcal{C}}
\newcommand{\R}{\mathbb{R}}
\newcommand{\N}{\mathbb{N}}
\newcommand{\D}{\mathcal{D}}
\newcommand{\OO}{\Omega}

\newcommand{\Sig}{\mathcal{S}}
\newcommand{\sat}{\vDash}

\newtheorem{defn}{Definition}
\newtheorem{thm}{Theorem}

\newtheorem{prop}{Proposition}
\newtheorem{cor}{Corollary}
\newtheorem{remark}{Remark}

\author{Miel Sharf, Bart Besselink and Karl Henrik Johansson
\thanks{M. Sharf and K. H. Johansson are with the Division of Decision and Control Systems, KTH Royal Institute of
Technology, 10044 Stockholm, Sweden (e-mail: {\tt\small \{sharf,kallej\}@kth.se}). B. Besselink is with the Bernoulli Institute for Mathematics, Computer Science and
Artificial Intelligence, University of Groningen, 9700 AK Groningen, The Netherlands (e-mail: {\tt\small b.besselink@rug.nl})}%
}

\begin{document}

\maketitle
\thispagestyle{empty}
\pagestyle{empty}

\begin{abstract}
Verifying specifications for large-scale modern engineering systems can be a time-consuming task, as most formal verification methods are limited to systems of modest size. Recently, contract-based design and verification has been proposed as a modular framework for specifications, and linear-programming-based techniques have been presented for verifying that a given system satisfies a given contract. In this paper, we extend this assume/guarantee framework by presenting necessary and sufficient conditions for a collection of contracts on individual components to refine a contract on the composed system. These conditions can be verified by solving linear programs, whose number grows linearly with the number of specifications defined by the contracts. We exemplify the tools developed using a case study considering safety in a car-following scenario, where noise and time-varying delay are considered.
\end{abstract}

\section{Introduction}
Modern engineering systems tend to be large-scale and multi-disciplinary systems, comprised of components of different types and functions, including planning, sensing, controlling, actuating, and timing. Such systems are subject to specifications of different types, including performance and safety. However, many systems, e.g., intelligent transportation systems and smart manufacturing systems, are subject to complex specifications that cannot be captured using the classical tools of dissipativity theory \cite{vanderSchaft2000} or set invariance \cite{Blanchini2008}. These include behaviour, tracking, timing, and temporal logic specifications. 

Formal methods in control have been developed to capture these more complex specifications, typically defined via a temporal logic formula on the state of the system \cite{Belta2007,Tabuada2009,Wongpiromsarn2010}. These methods rely on discretizing the dynamical system and studying the resulting discrete automaton. As a result, these methods are susceptible to the curse of dimensionality, limiting their application to systems of moderate size. Other approaches for verification rely on data, but basic probabilistic arguments show that much data is needed in order to verify any specification with high probability \cite{Shalev2017}. These problems become much more exaggerated when we wish to make small changes to the dynamical system (e.g., replacing a component with a comparable alternative), forcing us to start the verification process from scratch.

Contract-based design were suggested to tackle this problem. Contract theory was first introduced as a modular approach for software design \cite{Meyer1992}, and it has proved useful for design of cyber-physical methods, both in theory and in practice \cite{Nuzzo2014,Nuzzo2015}. Contracts prescribe assumptions on the environments a software component can act in, and guarantees on its behaviour in these environments \cite{Benveniste2018}.

Recently, several papers have tried to define assume/guarantee contracts for dynamical (control) systems, which must be accompanied by efficient computational tools for satisfaction, composition and refinement, allowing their application to real-world systems. The paper \cite{Besselink2019} presented a simulation-based framework for assume/guarantee contracts for continuous-time dynamical systems, where satisfaction can be verified using methods from geometric control theory \cite{vanderSchaft2004}. For discrete-time systems, several contract-based approaches were considered in literature, prescribing assumptions on the input signals and guarantees on the state and output signals \cite{Saoud2018,Saoud2019,Ghasemi2020,Eqtami2019}. Another approach was presented in \cite{Nuzzo2015,SharfADHS2020}, prescribing assumptions on the input and guarantees on the output, relative to the input. These allow one to consider sensors, observers, and specifications such as tracking, for which the guarantee on the output should depend on the input. Moreover, the framework allows preliminary analysis of systems before the state of the system (or even its size) has been determined, e.g. before we have determined whether a specific controller would be proportional or a PI controller. The paper \cite{SharfADHS2020} prescribed efficient linear programming (LP) based tools for verifying that a given LTI system satisfies a given contract, as well as similar tools for verifying that a given contract refines another contract. 

However, no LP-based tools are known for verifying compositional refinement, i.e. that a collection of contracts on individual components refine a contract on the composite system.  In this paper, we present an LP-based method for verifying that a composition of two contracts on components in cascade refines a contract on the composite system. Our methods are exact, in the sense that they give a necessary and sufficient condition for compositional refinement. Moreover, our methods can be extended to consider a larger cascade of systems, with only minor changes. We also demonstrate their applicability in a case study, in which we consider safety of autonomous vehicles under time-varying delay and noisy measurements.

The rest of the paper is organized as follows. Section \ref{sec.AG}
presents some basic notions from the assume/guarantee framework of \cite{SharfADHS2020}. Section \ref{sec.Comp} presents the LP-based methods for verifying compositional refinement, constituting the main contribution of this paper. Section \ref{sec.Num} provides a computational case study of a two-vehicle system with time-varying delay and noise.

\paragraph*{Notation}
The set of natural numbers is denoted as $\N = \{0,1,2,\ldots\}$. For two sets $X,Y$, we denote their Cartesian product by $X\times Y$. The Euclidean space of dimension $n$ will be denoted as $\R^n$. For a some $n\in \N$, the set of all discrete-time signals $\N\to \R^n$ is denoted as $\Sig^n$. For vectors $v,u \in \mathbb{R}^n$ , we understand $v \le u$ as an entry-wise inequality.

\section{Assume/Guarantee Contracts} \label{sec.AG}
In this section, we define the class of systems for which
we define an abstract framework of assume/guarantee
contracts, as well as supporting notions such as satisfaction, refinement, and cascaded composition. These notions were previously defined in \cite{SharfADHS2020}, see also \cite{Nuzzo2015}.

\begin{defn}
A dynamical system $\Sigma$ with an input in $\R^{n_d}$ and output in $\R^{n_y}$ is a set valued map $\Sig^{n_d}\rightrightarrows\Sig^{n_y}$.
\end{defn}
In other words, for any input signal $d(\cdot) \in \Sig^{n_d}$, $\Sigma(d)\subseteq \Sig^{n_y}$ is a set consisting of all possible corresponding outputs, see Fig. \ref{fig.DynSys}. 
\begin{figure}[b]
    \centering
    \vspace{-10pt}
    \includegraphics[scale = 0.5]{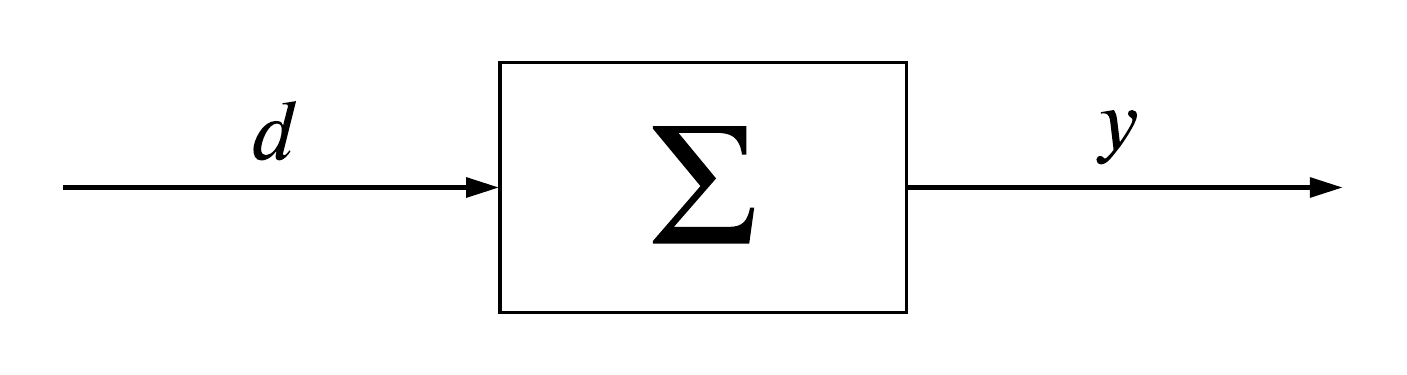}
    \vspace{-10pt}
    \caption{A dynamical system.}
    \label{fig.DynSys}
\end{figure}
We adopt the formulation of assume/guarantee contracts by prescribing assumptions on the input $d(\cdot)$ and demanding guarantees on the output $y(\cdot)$ given the input $d(\cdot)$.

\begin{defn}
An assume/guarantee contract with input $d(\cdot)\in \Sig^{n_d}$ and output $y(\cdot)\in \Sig^{n_y}$ is a pair $(\D,\OO)$ where $\D \subseteq \Sig^{n_d}$ are the assumptions and $\OO \subseteq \Sig^{n_d} \times \Sig^{n_y}$ are the guarantees. The contract $\C$ is called \emph{consistent} if for any $d(\cdot)\in \D$, there exists some $y(\cdot)$ such that $(d(\cdot),y(\cdot))\in \OO$.
\end{defn} 
In other words, we put assumptions on the input $d(\cdot)$ and demand guarantees on the input-output pair $(d(\cdot),y(\cdot))$. Assume/guarantee contracts define specifications on systems in the following manner:
\begin{defn}
We say that a system $\Sigma$ satisfies $\C = (\D,\OO)$ (or implements $\C$), and write $\Sigma \sat \C$, if for any $d(\cdot)\in \D$ and any $y(\cdot)\in \Sigma(d)$, we have $(d(\cdot),y(\cdot))\in \OO$.
\end{defn}

\subsection{Refinement and Composition}
The most important aspect of contract theory is its modularity, supported by the notions of refinement and composition. These allow one to replace a contract on a composite system by ``smaller" contracts on subsystems, which can be further replaced by even ``smaller" contracts on individual components. We start by defining refinement, dictating when one contract is stricter than another:
\begin{defn}
Let $\C_i = (\D_i,\OO_i)$ be contracts for $i=1,2$ with input in $\Sig^{n_d}$ and output in $\Sig^{n_y}$. We say $\C_1$ \emph{refines} $\C_2$ if $\D_1 \supseteq \D_2$ and $\Omega_1 \cap (\D_2 \times \Sig^{n_y}) \subseteq \Omega_2 \cap(\D_2 \times \Sig^{n_y})$. In that case, we write $\C_1 \preccurlyeq \C_2$.
\end{defn}
Colloquially, $\C_1 \preccurlyeq \C_2$ if $\C_1$ assumes less than $\C_2$, but guarantees more given the assumptions. The following proposition summarizes two important properties about refinement:
\begin{prop}[\cite{SharfADHS2020}]\label{prop.PropertiesOfRefinement}
Let $\C_i$ be assume/guarantee contracts for $i=1,2,3$, and let $\Sigma$ be a system. Then:
\begin{itemize}
    \item If $\C_1 \preccurlyeq \C_2$ and $\C_2 \preccurlyeq \C_3$ then $\C_1 \preccurlyeq \C_3$.
    \item If $\C_1 \preccurlyeq \C_2$ and $\Sigma \sat \C_1$, then $\Sigma \sat \C_2$.
\end{itemize}
\end{prop}

We now move to cascaded composition. Consider the block diagram in Fig. \ref{fig.Cascade}. 
\begin{figure}[t]
    \centering
    \includegraphics[width=0.45\textwidth]{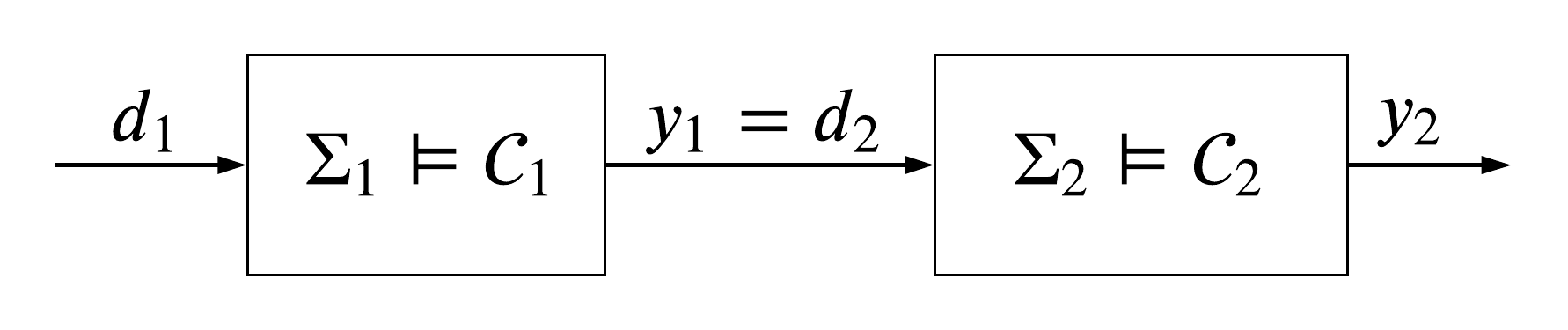}
    \vspace{-10pt}
    \caption{Cascade of contracts.}
    \vspace{-10pt}
    \label{fig.Cascade}
\end{figure}
Given two systems $\Sigma_1,\Sigma_2$ as in Fig. \ref{fig.Cascade}, their cascade $\Sigma_1\otimes\Sigma_2$ is defined via:
\begin{align*}
    \Sigma_1\otimes\Sigma_2(d) = \{y(\cdot):\ \exists z\in \Sigma_1(d)\text{~s.t.~} y\in\Sigma_2(z)\}.
\end{align*}
For two contracts, \cite{SharfADHS2020} defines their cascaded composition:
\begin{defn} \label{def.Composition}
Let $\C_i = (\D_i,\OO_i)$ be two contracts for $i=1,2$ as in Fig. \ref{fig.Cascade}. The cascaded composition $\C_1\otimes\C_2 = (\D_\otimes,\OO_\otimes)$ is a contract with input $d_\otimes = d_1$, output $y_\otimes = y_2$,
\begin{align*}
    \D_\otimes &= \{d_\otimes:~ d_\otimes\in \D_1, \left((d_\otimes,y_1)\in \OO_1 \implies y_1 \in \D_2\right)\},\\
    \OO_\otimes &= \{(d_\otimes,y_\otimes): \exists d_2 = y_1, (d_\otimes,y_1)\in \OO_1, (d_2,y_\otimes) \in \OO_2\}.
\end{align*}
\end{defn}
\vspace{5pt}
\begin{prop}[\cite{SharfADHS2020}] \label{prop.CompositionSat}
Let $\C_1,\C_2$ and $\Sigma_1,\Sigma_2$ be contracts and systems as in Fig. \ref{fig.Cascade}. If $\Sigma_1 \sat \C_1$ and $\Sigma_2 \sat \C_2$, then $\Sigma_1\otimes \Sigma_2 \sat \C_1 \otimes \C_2$.
\end{prop}

Propositions \ref{prop.PropertiesOfRefinement} and \ref{prop.CompositionSat} play a key role in contract theory, as they allow one to design a composite system to satisfy a given contract by designing each subsystem to satisfy a corresponding ``smaller" contract. This can be achieved by designing each component to satisfy a corresponding refined contract. In order to apply this modular design approach, we need efficient methods to verify both satisfaction $\Sigma \sat \C$ and composed refinement $\C_1 \otimes \C_2 \preccurlyeq \C$. The paper \cite{SharfADHS2020} provides efficient LP-based methods for verifying that a given linear time-invariant system $\Sigma$ satisfies a linear contract $\C$ (see Definition \ref{def.LinCon}). 
In this paper, we present efficient LP-based methods for verifying that $\C_1\otimes \C_2 \preccurlyeq \C$. We do so for contracts defined by linear inequalities:
\begin{defn}\label{def.LinCon}
A \emph{linear} contract $\C = (\D,\OO)$ is given by matrices $\mathfrak A^1,\mathfrak A^0,\mathfrak G^1,\mathfrak G^0$ and vectors $\mathfrak a^0,\mathfrak g^0$ of appropriate dimensions such that:
\begin{align*}
    \D &= \{d(\cdot): \mathfrak A^1 d(k+1) + \mathfrak A^0 d(k) \le \mathfrak a^0,~\forall k\},\\
    \OO &= \left\{(d(\cdot),y(\cdot)): \mathfrak G^1 \left[\begin{smallmatrix} d(k+1) \\ y(k+1) \end{smallmatrix}\right] + \mathfrak G^0 \left[\begin{smallmatrix} d(k) \\ y(k) \end{smallmatrix}\right] \le \mathfrak g^0,~ \forall k\right\}
\end{align*}
\end{defn}
\begin{remark}
Definition \ref{def.LinCon} only considers contracts with ``first-order" specifications, in the sense that the set of possible values of the signal at time $k$ is dictated only by its value at time $k-1$. The ideas and results in this paper can be extended to ``higher-order" specifications, in which the set of possible values of the signal at time $k$ is dictated by its value at times $k-1,k-2,\ldots,k-m$ for some $m\in \N$. We chose to restrict our attention to ``first-order" specifications in order to ease to notation and improve readability.
\end{remark}
\section{Verifying Compositional Refinement} \label{sec.Comp}
In this section, we provide computational tools for verifying that $\C_1 \otimes \C_2 \preccurlyeq \C$ for contracts $\C_1=(\D_1,\OO_1),\C_2=(\D_2,\OO_2)$ in cascade, as in Fig. \ref{fig.Cascade}, and a contract $\C=(\D,\OO)$ on the composite system. For simplicity of notation and in order to avoid confusion, we denote the input to both $\C,\C_1$ by $d(\cdot)$, the output from $\C,\C_2$ by $y(\cdot)$, and the intermediate signal by $z(\cdot)$. We assume throughout this section that $\C,\C_1,\C_2$ are all linear contracts, and write:
\begin{align} \label{eq.ThreeContracts}
    \D_1 &= \{d(\cdot): \mathfrak A^1 d(k+1) + \mathfrak A^0 d(k) \le \mathfrak a^0,~\forall k\},\\ \nonumber
    \OO_1 &= \left\{(d(\cdot),z(\cdot)): \mathfrak G^1 \left[\begin{smallmatrix} d(k+1) \\ z(k+1) \end{smallmatrix}\right] + \mathfrak G^0 \left[\begin{smallmatrix} d(k) \\ z(k) \end{smallmatrix}\right] \le \mathfrak g^0,~ \forall k\right\},\\\nonumber
    \D_2 &= \{z(\cdot): \mathfrak B^1 z(k+1) + \mathfrak B^0 z(k) \le \mathfrak b^0,~\forall k\},\\\nonumber
    \OO_2 &= \left\{(z(\cdot),y(\cdot)): \mathfrak H^1 \left[\begin{smallmatrix} z(k+1) \\ y(k+1) \end{smallmatrix}\right] + \mathfrak H^0 \left[\begin{smallmatrix} z(k) \\ y(k) \end{smallmatrix}\right] \le \mathfrak h^0,~ \forall k\right\},\\\nonumber
    \D &= \{d(\cdot): \mathfrak C^1 d(k+1) + \mathfrak C^0 d(k) \le \mathfrak c^0,~\forall k\},\\\nonumber
    \OO &= \left\{(d(\cdot),y(\cdot)): \mathfrak J^1 \left[\begin{smallmatrix} d(k+1) \\ y(k+1) \end{smallmatrix}\right] + \mathfrak J^0 \left[\begin{smallmatrix} d(k) \\ y(k) \end{smallmatrix}\right] \le \mathfrak j^0,~ \forall k\right\}.
\end{align}
Moreover, we denote $\C_1 \otimes \C_2 = \C_\otimes = (\D_\otimes,\OO_\otimes)$. Our goal is to give a computationally viable method of asserting that $\C_1\otimes \C_2 \preccurlyeq \C$, or equivalently, that $\D_\otimes \supseteq \D$ and $\OO_\otimes \cap (\D\times \Sig^{n_y}) \subseteq \OO \cap (\D\times \Sig^{n_y})$, see Fig. \ref{fig.CascadeRefinement}. These are equivalent to following implications:
\begin{itemize}
    \item If $d(\cdot) \in \D$, then $d(\cdot) \in \D_\otimes$.
    \item If $d(\cdot) \in \D$ and $(d(\cdot),y(\cdot)) \in \Omega_\otimes$, then $(d(\cdot),y(\cdot)) \in \Omega$.
\end{itemize}
In order to verify these implications, we use inductive reasoning, for which we need the following assumption, similarly to \cite{SharfADHS2020}:
\begin{defn}
Given two matrices $V^1,V^0$ and a vector $v^0$, we say $(V^1,V^0,v^0)$ is \emph{extendable} if for any two vectors $u_0,u_1$ and $V^1u_1 + V^0u_0 \le v^0$, there exists some vector $u_2$ such that $V^1 u_2 + V^0 u_1 \le v^0$.
\end{defn}
Assuming that $(\mathfrak A^1,\mathfrak A^0,\mathfrak a^0)$ (or any other triplet defining the contracts \eqref{eq.ThreeContracts}) is extendable is not very restrictive. It is equivalent to assuming that any signal $v(\cdot)$ adhering to the assumption, and defined for times $k=0,\ldots,n$, can be extended to a signal defined for all times $k\in \N$ while satisfying the assumption. Using the inductive reasoning framework, we prove the following theorem:

\begin{figure}[b]
    \centering
    \vspace{-15pt}
    \includegraphics[width=0.3\textwidth]{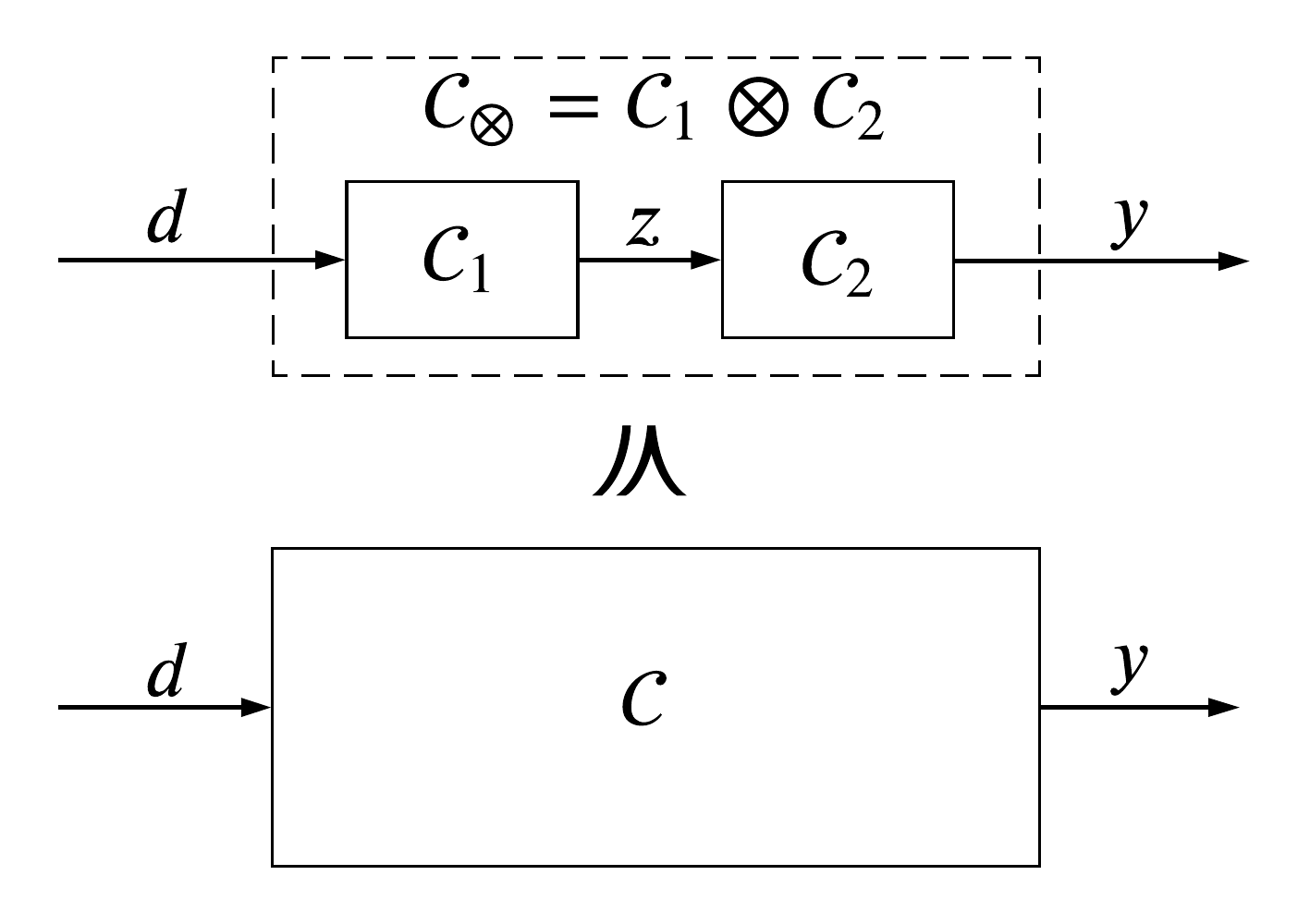}
    \caption{Cascaded Refinement, as in Theorem \ref{thm.Imp}}
        \vspace{-25pt}
    \label{fig.CascadeRefinement}
\end{figure}

\begin{thm} \label{thm.Imp}
Let $\C_1,\C_2,\C$ be contracts as in \eqref{eq.ThreeContracts}, and write $\mathfrak G^i = [\mathfrak G^i_d,\mathfrak G^i_z]$, $\mathfrak H^i = [\mathfrak H^i_z,\mathfrak H^i_y]$ for $i=0,1$. Assume that $(\mathfrak{C}^1,\mathfrak{C}^0,\mathfrak{c}^0)$, $\left(\left[\begin{smallmatrix}\mathfrak{C}^1 & 0 \\ \mathfrak{G}_d^1 & \mathfrak{G}_z^1 \end{smallmatrix}\right],\left[\begin{smallmatrix}\mathfrak{C}^0 & 0 \\ \mathfrak{G}_d^0 & \mathfrak{G}_z^0 \end{smallmatrix}\right],\left[\begin{smallmatrix} \mathfrak c^0 \\ \mathfrak g^0 \end{smallmatrix}\right]\right)$ and $\left(\left[\begin{smallmatrix}\mathfrak{C}^1 & 0 & 0 \\ \mathfrak{G}_d^1 & \mathfrak{G}_z^1 & 0 \\ 0 & \mathfrak{H}^1_z & \mathfrak{H}^1_y \end{smallmatrix}\right] ,\left[\begin{smallmatrix}\mathfrak{C}^0 & 0 & 0 \\ \mathfrak{G}_d^0 & \mathfrak{G}_z^0 & 0 \\ 0 & \mathfrak{H}^0_z &\mathfrak{H}^1_y\end{smallmatrix}\right] ,\left[\begin{smallmatrix} \mathfrak c^0 \\ \mathfrak g^0 \\ \mathfrak h^0 \end{smallmatrix}\right]\right)$ are all extendable.
Then, $\C_1 \otimes \C_2 \preccurlyeq \C$ if and only if the following three implications hold for any $d_0,d_1,z_0,z_1,y_0,y_1$:
\begin{enumerate}
    \item[i)] If $\mathfrak C^1 d_1 + \mathfrak C^0 d_0 \le \mathfrak c^0$, then $\mathfrak A^1 d_1 + \mathfrak A^0 d_0 \le \mathfrak a^0$.
    \item[ii)] If $\mathfrak C^1 d_1 + \mathfrak C^0 d_0 \le \mathfrak c^0$ and $\mathfrak G^1 \left[\begin{smallmatrix} d_1 \\ z_1 \end{smallmatrix}\right] +\mathfrak G^0 \left[\begin{smallmatrix} d_0 \\ z_0 \end{smallmatrix}\right] \le\mathfrak g^0$, then $\mathfrak B^1 z_1 + \mathfrak B^0 z_0 \le \mathfrak b^0$.
    \item[iii)] If $\mathfrak C^1 d_1 + \mathfrak C^0 d_0 \le \mathfrak c^0$, $\mathfrak G^1 \left[\begin{smallmatrix} d_1 \\ z_1 \end{smallmatrix}\right] + \mathfrak G^0 \left[\begin{smallmatrix} d_0 \\ z_0 \end{smallmatrix}\right] \le \mathfrak g^0$, and $\mathfrak H^1 \left[\begin{smallmatrix} z_1 \\ y_1 \end{smallmatrix}\right] + \mathfrak H^0 \left[\begin{smallmatrix} z_0 \\ y_0 \end{smallmatrix}\right] \le \mathfrak h^0$, then $\mathfrak J^1 \left[\begin{smallmatrix} d_1 \\ y_1 \end{smallmatrix}\right] + \mathfrak J^0 \left[\begin{smallmatrix} d_0 \\ y_0 \end{smallmatrix}\right] \le \mathfrak j^0$.
\end{enumerate}
\end{thm}

\begin{proof}
We prove that, under the extendability assumptions of the theorem, the first two implications are equivalent to $\D \subseteq \D_\otimes$, and the third is equivalent to $\OO_\otimes \cap (\D\times \Sig^{n_y}) \subseteq \OO \cap (\D\times \Sig^{n_y})$. We start with the former equivalence.

Suppose first that the implications i) and ii) hold. We take $d(\cdot)\in \D$ and show that $d(\cdot)\in \D_\otimes$. As $d(\cdot)\in \D_1$, we have $\mathfrak C^1 d(k+1) + \mathfrak C^0 d(k) \le \mathfrak c^0$ for any $k\in \N$, meaning that $\mathfrak A^1 d(k+1) + \mathfrak A^0 d(k) \le \mathfrak a^0$ holds by i). Moreover, we take some $z(\cdot)\in \Sig^{n_z}$ such that $(d(\cdot),z(\cdot))\in \Omega_1$, and claim that $z(\cdot)\in \D_2$. Indeed, for any $k\in \N$, the two inequalities:

\vspace{-8pt}
\small
\begin{align*}
    \mathfrak C^1 d(k+1) + \mathfrak C^0 d(k) \le \mathfrak c^0,~
    \mathfrak G^1 \begin{bmatrix} d(k+1) \\ z(k+1) \end{bmatrix} +\mathfrak G^0 \begin{bmatrix} d(k) \\ z(k) \end{bmatrix} \le\mathfrak g^0
\end{align*}\normalsize
imply that $\mathfrak B^1 z(k+1) + \mathfrak B^0 z(k) \le \mathfrak b^0$ holds by ii). Combining both claims, we conclude that $d\in \D_\otimes$.

Conversely, we assume $\D_\otimes \supseteq \D$, and show that both implications i) and ii) hold. Suppose first that $d_0,d_1\in \R^{n_d}$ satisfy $\mathfrak C^1 d_1 + \mathfrak C^0 d_0 \le \mathfrak c^0$. By extendibility, we can find a signal $d(\cdot)\in \D$ such that $d(0)=d_0,d(1)=d_1$. As $\D_\otimes \supseteq \D$, we conclude that $d(\cdot)\in \D_1$, so $\mathfrak A^1 d_1 + \mathfrak A^0 d_0 \le \mathfrak a^0$ is achieved at time $k=0$. Now, assume that $d_0,d_1,y_0,y_1$ satisfy both inequalities:
\begin{align*}
    \mathfrak C^1 d_1 + \mathfrak C^0 d_0 \le \mathfrak c^0,~
    \mathfrak G^1 \begin{bmatrix} d_1 \\ z_1 \end{bmatrix} +\mathfrak G^0 \begin{bmatrix} d_0 \\ z_0 \end{bmatrix} \le\mathfrak g^0
\end{align*}
which can be concisely written as:
\begin{align*}
    \begin{bmatrix}
    \mathfrak{C}^1 & 0 \\ \mathfrak{G}_d^1 & \mathfrak{G}_z^1
    \end{bmatrix}
    \begin{bmatrix}
    d_1 \\ z_1
    \end{bmatrix} + 
    \begin{bmatrix}
    \mathfrak{C}^0 & 0 \\ \mathfrak{G}_d^0 & \mathfrak{G}_z^0
    \end{bmatrix}
    \begin{bmatrix}
    d_0 \\ z_0
    \end{bmatrix} \le 
    \begin{bmatrix}
    \mathfrak c^0 \\ \mathfrak g^0
    \end{bmatrix} 
\end{align*}
By extendibility, we find a signal $\mu(\cdot) = [d(\cdot)^\top ,z(\cdot)^\top]^\top\in \Sig^{n_d+n_z}$ such that $d(0)=d_0,d(1)=d_1,z(0)=z_0,z(1)=z_1$, and for all $k\in \N$, we have:
\begin{align*}
    \begin{bmatrix}
    \mathfrak{C}^1 & 0 \\ \mathfrak{G}_d^1 & \mathfrak{G}_z^1
    \end{bmatrix}
    \begin{bmatrix}
    d(k+1) \\ z(k+1)
    \end{bmatrix} + 
    \begin{bmatrix}
    \mathfrak{C}^0 & 0 \\ \mathfrak{G}_d^0 & \mathfrak{G}_z^0
    \end{bmatrix}
    \begin{bmatrix}
    d(k) \\ z(k)
    \end{bmatrix} \le 
    \begin{bmatrix}
    \mathfrak c^0 \\ \mathfrak g^0
    \end{bmatrix} 
\end{align*}
Thus, $d(\cdot)\in \D\subseteq \D_\otimes$ and $(d(\cdot),z(\cdot))\in \Omega_1$, and we conclude that $z(\cdot)\in \D_2$ by the definition of $\D_\otimes$. Hence, for time $k=0$, we get $\mathfrak B^1 z_1 + \mathfrak B^0 z_0 \le \mathfrak b^0$, as desired.
\\\\
We now show that the implication iii) is equivalent to $\OO_\otimes \cap (\D\times \Sig^{n_y}) \subseteq \OO \cap (\D\times \Sig^{n_y})$. Suppose first that iii) holds, and show that the inclusion holds. Take any $(d(\cdot),y(\cdot)) \in \OO_\otimes \cap (\D\times \Sig^{n_y})$, so that $d(\cdot)\in \D$ and $(d(\cdot),y(\cdot))\in \OO_\otimes$. By Definition \ref{def.Composition}, there exists some signal $z(\cdot)$ such that $(d(\cdot),z(\cdot))\in \Omega_1$ and $(z(\cdot),y(\cdot)) \in \Omega_2$. Thus, the following inequalities hold for all $k\in \N$:

\vspace{-8pt}
\begin{align*}
    &\mathfrak C^1 d(k+1) + \mathfrak C^0 d(k) \le \mathfrak c^0,\\
    &\mathfrak G^1 \begin{bmatrix} d(k+1) \\ z(k+1) \end{bmatrix} + \mathfrak G^0 \begin{bmatrix} d(k) \\ z(k) \end{bmatrix} \le \mathfrak g^0,\\
    &\mathfrak H^1 \begin{bmatrix} z(k+1) \\ y(k+1) \end{bmatrix} + \mathfrak H^0 \begin{bmatrix} z(k) \\ y(k) \end{bmatrix} \le \mathfrak h^0.
\end{align*}
We use iii) and conclude that for any $k\in\N$, we have:
\begin{align*}
    \mathfrak J^1 \begin{bmatrix} d(k+1) \\ y(k+1) \end{bmatrix} + \mathfrak J^0 \begin{bmatrix} d(k) \\ y(k) \end{bmatrix} \le \mathfrak j^0
\end{align*}
In particular, $(d(\cdot),y(\cdot)) \in \OO$, as desired.

Conversely, assume that $\OO_\otimes \cap (\D\times \Sig^{n_y}) \subseteq \OO \cap (\D\times \Sig^{n_y})$. Let $d_0,d_1,z_0,z_1,y_0,y_1$ be such that the three inequalities
\begin{align*}
    &\mathfrak C^1 d_1 + \mathfrak C^0 d_0 \le \mathfrak c^0, \\
    &\mathfrak G^1 \begin{bmatrix} d_1 \\ z_1 \end{bmatrix} + \mathfrak G^0 \begin{bmatrix} d_0 \\ z_0 \end{bmatrix} \le \mathfrak g^0,\\
    &\mathfrak H^1 \begin{bmatrix} z_1 \\ y_1 \end{bmatrix} + \mathfrak H^0 \begin{bmatrix} z_0 \\ y_0 \end{bmatrix} \le \mathfrak h^0,
\end{align*}
all hold. We can write these inequalities compactly as:
\begin{align*}
    \begin{bmatrix}
    \mathfrak{C}^1 & 0 & 0 \\ \mathfrak{G}_d^1 & \mathfrak{G}_z^1 & 0 \\ 0 & \mathfrak{H}^1_z & \mathfrak{H}^1_y
    \end{bmatrix}
    \begin{bmatrix}
    d_1 \\ z_1 \\ y_1
    \end{bmatrix} + 
    \begin{bmatrix}
    \mathfrak{C}^0 & 0 & 0 \\ \mathfrak{G}_d^0 & \mathfrak{G}_z^0 & 0 \\ 0 & \mathfrak{H}^0_z & \mathfrak{H}^0_y    \end{bmatrix}
    \begin{bmatrix}
    d_0 \\ z_0 \\ y_0
    \end{bmatrix} \le 
    \begin{bmatrix}
    \mathfrak c^0 \\ \mathfrak g^0 \\ \mathfrak h^0
    \end{bmatrix} 
\end{align*}
By extendibility, we find $\kappa(\cdot) = [d(\cdot)^\top ,z(\cdot)^\top,y(\cdot)^\top]^\top\in \Sig^{n_d+n_z+n_y}$ such that $d(i)=d_i,z(i)=z_i$ and $y(i)=y_i$ all hold for $i=0,1$, and for all $k\in \N$, we have:

\vspace{-8pt}
\footnotesize
\begin{align*}
    \begin{bmatrix}
    \mathfrak{C}^1 & 0 & 0 \\ \mathfrak{G}_d^1 & \mathfrak{G}_z^1 & 0 \\ 0 & \mathfrak{H}^1_z & \mathfrak{H}^1_y
    \end{bmatrix}
    \begin{bmatrix}
    d(k+1) \\ z(k+1) \\ y(k+1)
    \end{bmatrix} + 
    \begin{bmatrix}
    \mathfrak{C}^0 & 0 & 0 \\ \mathfrak{G}_d^0 & \mathfrak{G}_z^0 & 0 \\ 0 & \mathfrak{H}^0_z & \mathfrak{H}^0_y    \end{bmatrix}
    \begin{bmatrix}
    d(k) \\ z(k) \\ y(k)
    \end{bmatrix} \le 
    \begin{bmatrix}
    \mathfrak c^0 \\ \mathfrak g^0 \\ \mathfrak h^0
    \end{bmatrix} 
\end{align*}
\normalsize
Thus, we conclude that $d(\cdot)\in \D$, $(d(\cdot),z(\cdot))\in \Omega_1$ and $(z(\cdot),y(\cdot))\in \Omega_2$. In particular, by Definition \ref{def.LinCon}, we have that $(d(\cdot),y(\cdot)) \in \Omega_\otimes \cap (\D\times \Sig^{n_y})$. In turn, we conclude that implies that $(d(\cdot),y(\cdot)) \in \Omega \cap (\D\times \Sig^{n_y})$, and for time $k=0$, we yield 
\vspace{-5pt}
\begin{align*}
    \mathfrak J^1 \begin{bmatrix} d_1 \\ y_1 \end{bmatrix} + \mathfrak J^0 \begin{bmatrix} d_0 \\ y_0 \end{bmatrix} \le \mathfrak j^0
\end{align*}
This completes the proof.
\end{proof}

\begin{remark} \label{rem.Delay}
Theorem \ref{thm.Imp} assumes that the contracts in \eqref{eq.ThreeContracts} define the assumptions and guarantees for all times $k\in \N$, and in particular, for the same times $k$. The case study in section \ref{sec.Num} includes a contract on a delayed observer, for which the input specifications are defined for all $k\ge 0$, but the output specifications are only defined for times $k\ge 1$. This is the case whenever $\mathfrak G^0_y = 0$. In this case, the theorem is amended by considering implications for longer parts of the signal. For example, the implication ii) would be replaced by the following: If $\mathfrak C^1 d_{i+1} + \mathfrak C^0 d_i \le \mathfrak c^0$ and $\mathfrak G^1 \left[\begin{smallmatrix} d_{i+1} \\ z_{i+1} \end{smallmatrix}\right] +\mathfrak G^0 \left[\begin{smallmatrix} d_i \\ z_i \end{smallmatrix}\right] \le\mathfrak g^0$ hold for $i=0,1$, then $\mathfrak B^1 z_2 + \mathfrak B^0 z_1 \le \mathfrak b^0$. Implication iii) is similarly treated.
\end{remark}

Theorem \ref{thm.Imp} allows one to verify that a given composition of cascaded contracts refines a contract on the composite system by proving three implications, each of them can be cast as an optimization problem:
\begin{prop} \label{prop.Opt}
Suppose that the assumptions of Theorem \ref{thm.Imp} hold. The refinement $\C_1 \otimes \C_2 \preccurlyeq \C$ holds if and only if $\varrho_\D,\varrho_\otimes,\varrho_\OO$, which are defined using the following optimization problems, are all non-positive:
\begin{align*}
    \varrho_\D = \max ~&~ \max_j \left[{\rm e_j}^\top \left(\mathfrak A^1 d_1 + \mathfrak A^0 d_0 - \mathfrak a^0\right)\right]\\
    {\rm s.t.}~&~ \mathfrak C^1 d_1 + \mathfrak C^0 d_0 \le \mathfrak c^0,\\
    ~&~ d_0,d_1 \in \R^{n_d},
\end{align*}
\begin{align*}
    \varrho_\otimes = \max ~&~ \max_j \left[{\rm e_j}^\top \left(\mathfrak B^1 z_1 + \mathfrak B^0 z_0 - \mathfrak b^0\right)\right]\\
    {\rm s.t.}~&~ \mathfrak C^1 d_1 + \mathfrak C^0 d_0 \le \mathfrak c^0,\\
    ~&~\mathfrak G^1 \left[\begin{smallmatrix} d_1 \\ z_1 \end{smallmatrix}\right] +\mathfrak G^0 \left[\begin{smallmatrix} d_0 \\ z_0 \end{smallmatrix}\right] \le\mathfrak g^0\\
    ~&~ d_0,d_1 \in \R^{n_d},~z_0,z_1 \in \R^{n_z},
\end{align*}
\begin{align*}
    \varrho_\OO = \max ~&~ \max_j \left[{\rm e_j}^\top \left(\mathfrak J^1 \left[\begin{smallmatrix} d_1 \\ y_1 \end{smallmatrix}\right] + \mathfrak J^0 \left[\begin{smallmatrix} d_0 \\ y_0 \end{smallmatrix}\right] - \mathfrak j^0\right)\right]\\
    {\rm s.t.}~&~ \mathfrak C^1 d_1 + \mathfrak C^0 d_0 \le \mathfrak c^0,\\
    ~&~\mathfrak G^1 \left[\begin{smallmatrix} d_1 \\ z_1 \end{smallmatrix}\right] +\mathfrak G^0 \left[\begin{smallmatrix} d_0 \\ z_0 \end{smallmatrix}\right] \le\mathfrak g^0\\
    ~&~\mathfrak H^1 \left[\begin{smallmatrix} z_1 \\ y_1 \end{smallmatrix}\right] + \mathfrak H^0 \left[\begin{smallmatrix} z_0 \\ y_0 \end{smallmatrix}\right] \le \mathfrak h^0\\
    ~&~ d_0,d_1 \in \R^{n_d},~z_0,z_1 \in \R^{n_z},~y_0,y_1\in \R^{n_y},
\end{align*}
where ${\rm e_i}$ are the standard basis vectors.
\end{prop}
\begin{proof}
Implication i) in Theorem \ref{thm.Imp} holds if and only if $\mathfrak C^1 d_1 + \mathfrak C^0 d_0 \le \mathfrak c^0$ implies $\mathfrak A^1 d_1 + \mathfrak A^0 d_0 - \mathfrak a^0 \le 0$, which holds if and only if $\varrho_\D \le 0$. Similarly, implication ii) holds if and only if $\varrho_{\otimes} \le 0$, and implication iii) in holds if and only if $\varrho_{\OO}\le 0$.
\end{proof}

Proposition \ref{prop.Opt} suggests a method for verifying whether $\C_1\otimes\C_2 \preccurlyeq \C$ is true or false by solving three optimization problems, and checking whether their optimal values are non-positive. These optimization problems can be recast as linear programs. Indeed, fixing some value of $j$ and removing the inner maximization results in a linear program. Denoting the value of these linear programs by $\vartheta_{\D,j},\vartheta_{\otimes,j},\vartheta_{\OO,j}$ respectively, we see that $\varrho_\D = \max_j \vartheta_{\D,j}, \varrho_\otimes = \max_j \vartheta_{\otimes,j}$ and $\varrho_\OO = \max_j \vartheta_{\OO,j}$. In particular, $\varrho_{\D},\varrho_{\otimes},\varrho_{\OO} \le 0$ if and only if $\vartheta_{\D,j},\vartheta_{\otimes,j},\vartheta_{\OO,j}\le 0$ for each admissible $j$. We conclude:

\begin{cor} \label{cor.LPopt}
Suppose that the assumptions of Theorem \ref{thm.Imp} hold, and let $s_{\mathfrak A},s_{\mathfrak B},s_{\mathfrak J}$ be the number of rows in $\mathfrak{A}^1,\mathfrak{B}^1$ and $\mathfrak{J}^1$, respectively. The refinement $\C_1 \otimes \C_2 \preccurlyeq \C$ holds if and only if for any $i=1,\ldots,s_{\mathfrak A}$, $j=1,\ldots,s_{\mathfrak B}$ and any $\ell = 1,\ldots,s_{\mathfrak J}$, we have $\vartheta_{\D,i},\vartheta_{\otimes,j},\vartheta_{\OO,\ell}\le 0$. In particular, deciding whether $\C_1\otimes \C_2 \preccurlyeq \C$ holds or not can be achieved by solving $s_{\mathfrak A}+s_{\mathfrak B} + s_{\mathfrak J}$ linear programs.
\end{cor}

To conclude this section, we presented a method for determining whether $\C_1 \otimes \C_2 \preccurlyeq \C$ holds, where $\C_1,\C_2,\C$ are of the form \eqref{eq.ThreeContracts}, by solving prescribed linear programs. These can be efficiently solved using off-the-shelf optimization software (e.g., Yalmip \cite{Lofberg2004}). Moreover, the number of optimization problems grows linearly with the number of rows in the matrices defining the contracts, which can be understood as the number of different half-space constraints required to define the assumptions and guarantees. Thus, the presented verification method is computationally viable.

\begin{figure}[b]
    \centering
    \includegraphics[width = 0.47\textwidth]{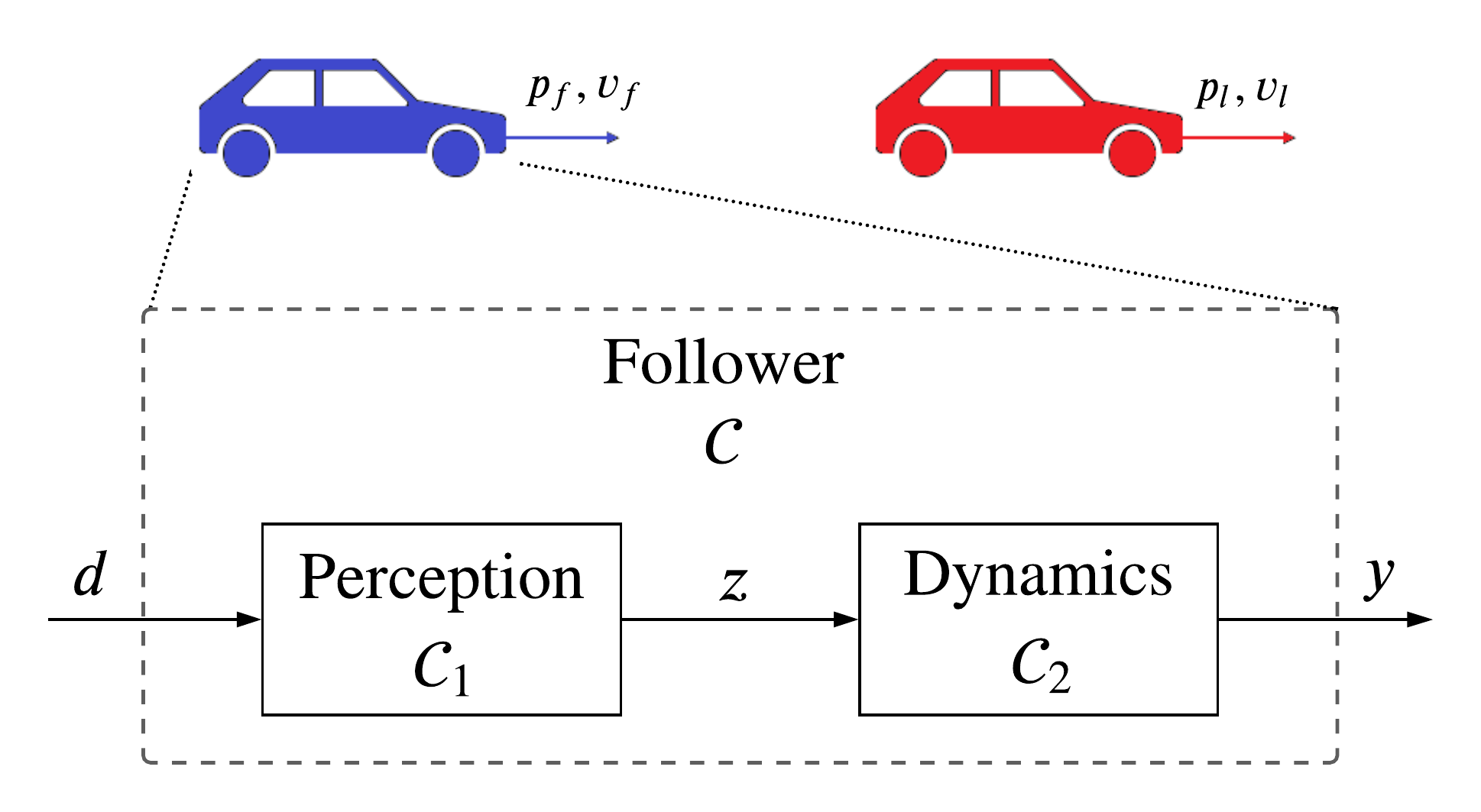}
    \caption{The two-vehicle scenario of Section \ref{sec.Num}.}
    \label{fig.TwoCar}
\end{figure}

\section{Numerical Example} \label{sec.Num}
We now demonstrate the results of the previous section using a case study regarding safety in a car-follower scenario.

\subsection{Problem Statement}
Consider two vehicles driving along a single-lane highway, as in Fig. \ref{fig.TwoCar}. We are given a headway $h>0$, and a specification that the follower keeps at least the given headway from the leader. Denoting the position and velocity of the follower by $p_f(k)$, $v_f(k)$, and the position and velocity of the leader by $p_l(k),v_l(k)$, the specification is that $p_l(k) - p_f(k) - hv_f(k) \ge 0$ holds at any time $k\in \N$. 

We will decompose the follower vehicle into two subsystems - a perception system, which provides noisy, delayed measurements of the leader vehicle, and a dynamics subsystem, driving the follower according to the provided measurements. We will first use the framework of assume/guarantee contracts and the tools provided in Section \ref{sec.Comp} to derive specifications on both subsystems. Later, use the framework of \cite{SharfADHS2020} to verify that a given model satisfies these specifications. Lastly, we demonstrate our findings using a simulation.

\subsection{Assume-Guarantee Specifications}
Let us first define the contract $\C$ on the composite system, i.e. the follower vehicle. The input $d(\cdot)$ is given by $d(k) = [p_l(k),v_l(k)]$, and we assume the leader satisfies the following kinematic relations for all $k\in \N$:
\begin{align} \label{eq.Kinematics}
&p_l(k+1) = p_l(k) + \Delta t v_l(k),\\ \nonumber &v_l(k+1) = v_l(k) + \Delta ta_l(k),\\ \nonumber
&a_l(k) \in [-a_{\rm min},a_{\rm max}]
\end{align}
where $a_l(k)$ is the acceleration of the leading vehicle and $\Delta t > 0$ is the length of a discrete time step. We also assume that $v_l(k) \ge 0$ at any time $k$. The output $y(\cdot)$ from the composite system is given by $y(k) = [p_f(k),v_f(k)]$, and we wish to guarantee that headway is kept, i.e. that $p_l(k) - p_f(k) - hv_f(k) \ge 0$ holds for any $k\in \N$. The contract defined above is linear and is given by the matrices:
\begin{align*}
&\mathfrak C^1 = \left[\begin{smallmatrix} 1 & 0 \\ -1 & 0 \\ 0 & 1 \\ 0 & -1 \\ 0 & 0 \end{smallmatrix}\right],
&&\mathfrak C^0 = \left[\begin{smallmatrix} -1 & -\Delta t \\ 1 & \Delta t \\ 0 & -1 \\ 0 & 1 \\ 0 & -1 \end{smallmatrix}\right]
&&\mathfrak c^0 = \left[\begin{smallmatrix} 0 \\ 0 \\ \Delta t a_{\rm max} \\ \Delta t a_{\rm min} \\ 0\end{smallmatrix}\right], \\
&\mathfrak J^1 = \left[\begin{smallmatrix} 0 & 0 & 0 & 0\end{smallmatrix}\right], 
&&\mathfrak J^0 = \left[\begin{smallmatrix} -1 & 0 & 1 & h\end{smallmatrix}\right], 
&&\mathfrak j^0 = [0].
\end{align*}
\\ \vspace{-15pt}

We now consider the different subsystems of the follower, as in Fig. \ref{fig.TwoCar}. We start by defining a contract $\C_1$ on the perception system. The input $d(\cdot)$ is the same as before, and the output $z(\cdot)$ is composed of measurements of the position and velocity of the leading vehicle, denoted as $z(k) = [p_m(k),v_m(k)]^\top$. We assume as before that the leader vehicle follows the laws of kinematics, i.e. that \eqref{eq.Kinematics} is satisfied, and that the velocity of the leader is always non-negative. As for guarantees, we allow the measurements to have delay and noise, up to a bounded amount. We assume that the delay is no larger than one time step $\Delta t$ and use a linear interpolation as a model for the delay. Namely, we denote the maximum allowable delay by $\tau\le \Delta t$, and the bounds on the allowable noise in the position and velocity measurements by $\delta_p,\delta_v$, and assume that for any time $k\in \N$,
\begin{align*}
    p_m(k+1) &= \left(1-\frac{\sigma_{k+1}^p}{\Delta t}\right) p_l(k+1) + \frac{\sigma^p_{k+1}}{\Delta t}p_l(k) + \nu_{p}(k)\\
    v_m(k+1) &=    \left(1-\frac{\sigma_{k+1}^v}{\Delta t}\right) v_l(k+1) + \frac{\sigma^v_{k+1}}{\Delta t}v_l(k) + \nu_{v}(k)
\end{align*}
where $\sigma^p_{k+1},\sigma^v_{k+1} \in [0,\tau]$ are the delays in the position and velocity measurements, respectively, $\nu_p(k),\nu_v(k)$ are the measurement noises, and $|\nu_p(k)|\le \delta_p, |\nu_v(k)|\le \delta_v$. We also guarantee that $v_m(k) \ge 0$ for all $k$. We show that this contract is linear. Using the assumptions, i.e. the kinematic relations \eqref{eq.Kinematics} and $v_l(k) \ge 0$, we restate these guarantees as:
\begin{align}\label{eq.pRestate}
    &p_m(k+1) = p_l(k+1) - s^p_{k+1} \Delta t v_l(k) + \nu_p(k)
    \\ \label{eq.vRestate}
    &v_m(k+1) = v_l(k+1) - s^v_{k+1} \Delta t a_l(k) + \nu_v(k)
\end{align}
where $s_{k+1}^p,s_{k+1}^v \in [0,\frac{\tau}{\Delta t}]$ are given by $s_{k+1}^p = \frac{\sigma^p_{k+1}}{\Delta t}$ and $s_{k+1}^v = \frac{\sigma^v_{k+1}}{\Delta t}$. Using the assumption $a_l(k) \in [-a_{\rm min},a_{\rm max}]$, \eqref{eq.vRestate} can be restated as a pair of inequalities:
\begin{align}\label{eq.vFinalForm}
    &v_m(k+1)\ge v_l(k+1)-\tau a_{\rm max} - \delta_v\\ \nonumber
    &v_m(k+1)\le v_l(k+1)+\tau a_{\rm min} + \delta_v,
\end{align}
and using the assumption $v_l(k)\ge 0$ allows one to restate \eqref{eq.pRestate} as the following pair of inequalities:
\begin{align}\label{eq.pFinalForm}
    &p_m(k+1) \ge p_l(k+1) - \tau v_l(k) - \delta_p\\ \nonumber
    &p_m(k+1) \le p_l(k+1) + \delta_p.
\end{align}
To conclude, we choose the contract $\C_1$ on the perception system with assumptions \eqref{eq.Kinematics} and $v_l(k)\ge 0$ for all $k\in \N$, and guarantees \eqref{eq.vFinalForm}, \eqref{eq.pFinalForm} and $v_m(k) \ge 0$ for all $k\in \N$. This contract is linear, and is given by the matrices:
\begin{align*}
&\mathfrak A^1 = \left[\begin{smallmatrix} 1 & 0 \\ -1 & 0 \\ 0 & 1 \\ 0 & -1 \\ 0 & 0 \end{smallmatrix}\right],
\mathfrak A^0 = \left[\begin{smallmatrix} -1 & -\Delta t \\ 1 & \Delta t \\ 0 & -1 \\ 0 & 1 \\ 0 & -1 \end{smallmatrix}\right],
\mathfrak a^0 = \left[\begin{smallmatrix} 0 \\ 0 \\ \Delta t a_{\rm max} \\ \Delta t a_{\rm min} \\ 0\end{smallmatrix}\right], \\
&\mathfrak G^1 = \left[\begin{smallmatrix} 1 & 0 & -1 & 0 \\ -1 & 0 & 1 & 0 \\ 0 & 1 & 0 & -1 \\ 0 & -1 & 0 & 1 \\ 0 & 0 & 0 & 0 \end{smallmatrix}\right], 
\mathfrak G^0 = \left[\begin{smallmatrix} 0 & -\tau & 0 & 0 \\ 0 & 0 & 0 & 0 \\ 0 & 0 & 0 & 0 \\  0 & 0 & 0 & 0 \\ 0 & 0 & 0 & -1
\end{smallmatrix}\right], 
\mathfrak g^0 = \left[\begin{smallmatrix}\delta_p \\ \delta_p \\ \mu_{\rm max} \\ \mu_{\rm min} \\ 0 \end{smallmatrix}\right].
\end{align*}
where $\mu_{\rm max} = \tau a_{\rm max} + \delta_v$ and $\mu_{\rm min} = \tau a_{\rm min} + \delta_v$. 
\\

Lastly, we consider the contract $\C_2$ on the desired (controlled) dynamics. The input $z(\cdot)$ to the dynamics is given by $z(k) = [p_m(k),v_m(k)]$, i.e. it consists of position and velocity measurements of the leading vehicle. The output $y(\cdot)$ is given by $y(k) = [p_f(k),v_f(k)]$, i.e. the position and velocity of the follower vehicle. We assume that the position and velocity measurements satisfy a noisy and delayed version of the kinematic relations:
\begin{align*}
    &p_m(k+1)-p_m(k)-\Delta tv_m(k) \ge  -\tau v_m(k)-\xi_{\rm down} \\
    &p_m(k+1)-p_m(k)-\Delta tv_m(k) \le  \tau v_m(k) + \xi_{\rm up} \\
    &v_m(k+1) - v_m(k) \in [-\eta_{\rm down},\eta_{\rm up}]
\end{align*}
where $\xi_{\rm down},\xi_{\rm up},\eta_{\rm down},\eta_{\rm up}$ are constants to be decided later. As for guarantees, we wish to guarantee a robustified version of headway-keeping, as we use a delayed and noisy measurement of the position of the leader vehicle instead of its true position. Specifically, we guarantee that $p_m(k) - p_f(k) - hv_f(k) \ge \delta_p$ holds for any time $k$. The contract $\C_2$ is also linear, and is given by the matrices:
\begin{align*}
&\mathfrak B^1 = \left[\begin{smallmatrix} 1 & 0 \\ -1 & 0 \\ 0 & 1 \\ 0 & -1  \end{smallmatrix}\right],
&&\mathfrak B^0 = \left[\begin{smallmatrix} -1 & -\Delta t-\tau \\ 1 & \Delta t-\tau \\ 0 & -1 \\ 0 & 1\end{smallmatrix}\right],
&&\mathfrak b^0 = \left[\begin{smallmatrix} \xi_{\rm up} \\ \xi_{\rm down} \\ \eta_{\rm up} \\ \eta_{\rm down}\end{smallmatrix}\right], \\
&\mathfrak H^1 = \left[\begin{smallmatrix} 0 & 0 & 0 & 0\end{smallmatrix}\right], 
&&\mathfrak H^0 = \left[\begin{smallmatrix} -1 & 0 & 1 & h\end{smallmatrix}\right], 
&&\mathfrak h^0 = [\delta_p].
\end{align*}

We wish to verify that $\C_1 \otimes \C_2 \preccurlyeq \C$. To do so, we choose the parameters $h = 2{\rm sec}$, $\Delta t = 0.3{\rm sec}, a_{\rm min} = a_{\rm max} = 9.8{\rm m/sec^2}$, $\tau = 0.1{\rm sec}$, $\xi_{\rm up} = 1.75{\rm m}$, $\xi_{\rm down} = 1.45{\rm m}$ and $\eta_{\rm up}=\eta_{\rm down} = 5.1{\rm m/sec}$, and use the computational framework of Section \ref{sec.Comp}, and specifically Proposition \ref{prop.Opt}. The problems defining $\varrho_\D,\varrho_\OO$ are achieved by plugging in the matrices $\mathfrak{A,B,C,G,H,J}$ in Proposition \ref{prop.Opt}, and the problem defining $\varrho_\otimes$ is defined using both Proposition \ref{prop.Opt} and Remark \ref{rem.Delay} as:
\begin{align*}
    \varrho_\otimes = \max ~&~ c_\otimes(p_m^2,v_m^2,p_m^1,v_m^1)\\
    {\rm s.t.}~&~ p_l^1 = p_l^0 + \Delta t v_l^0,~~p_l^2 = p_l^1 + \Delta t v_l^1\\
    ~&~ v_l^1 - v_l^0 \le \Delta t a_{\rm max},~~v_l^1 - v_l^0 \ge -\Delta t a_{\rm min},\\
    ~&~ v_l^2 - v_l^1 \le \Delta t a_{\rm max},~~v_l^2 - v_l^1 \ge -\Delta t a_{\rm min},\\
    ~&~ v_l^0,v_l^1,v_l^2 \ge 0\\
    ~&~v_l^1 - \tau a_{\rm max}-\delta_v \le v_m^1 \le v_l^1 + \tau a_{\rm min}+\delta_v,\\
    ~&~v_l^2 - \tau a_{\rm max}-\delta_v \le v_m^2 \le v_l^2 + \tau a_{\rm min}+\delta_v,\\
    ~&~ p_l^1 - \tau v_l^0 -\delta_p \le p_m^1 \le p_l^1 + \delta_p\\
    ~&~ p_l^2 - \tau v_l^1 -\delta_p \le p_m^2 \le p_l^2 + \delta_p\\
    ~&~ p_m^1,p_m^2,v_m^1,v_m^2,p_l^0,p_l^1,p_l^2,v_l^0,v_l^1,v_l^2 \in \R,
\end{align*}
where: \small
\begin{align*}
     c_\otimes(p_m^2,v_m^2,p_m^1,v_m^1) = \max\{&p_m^2 - p_m^1 - (\Delta t +\tau) v_m^1 - \xi_{\rm up},\\&p_m^1-p_m^2+(\Delta t - \tau)v_m^1 - \xi_{\rm down},\\&v_m^2 - v_m^1 - \eta_{\rm up},\\&v_m^1 - v_m^2 - \eta_{\rm down}\}.
\end{align*}\normalsize
We solve the three optimization problems using Yalmip \cite{Lofberg2004}, and find that $\varrho_\D = \varrho_\otimes = \varrho_\OO = 0$. Thus, we indeed have $\C_1\otimes \C_2 \preccurlyeq \C$ by Proposition \ref{prop.Opt}. As a result, we can implement each of the contracts $\C_1$, $\C_2$ independently of one another, and the resulting composition (i.e., model for the follower vehicle) will satisfy the contract $\C$.

\subsection{Implementation and Simulation}
Let us demonstrate the fact that $\C_1 \otimes \C_2 \preccurlyeq \C$ by choosing specific models for the perception and (controlled) dynamics subsystems, and showing that the guarantees on the composite systems are kept in a simulation. For the perception system, we choose the delay and noise as independent uniformly distributed random variables inside the relevant intervals, so the contract $\C_1$ is satisfied. 

For the (controlled) dynamics, we choose the system as a double integrator with an affine controller, given by the equations:
\begin{align*}
    p_f(k+1) = p_f(k) + \Delta t v_f(k),\\
    v_f(k+1) = v_f(k) + \Delta t a_f(k),
    \end{align*}
where at each time $k$, the control input $a_f$ is given as:
    \begin{align*}
    a_f = \frac{p_{m}-p_{f}-hv_{f}}{h\Delta t}+\frac{(\Delta t-\tau)v_{m}}{h\Delta t}-\frac{1}{h}v_{f}-\frac{\lambda}{h\Delta t},
\end{align*}
for $\lambda = \xi_{\rm down} + \delta_p$. We choose the initial conditions $[p_f(0),v_f(0)]$ so that $p_m(0) - p_f(0) - hv_f(0) \ge \delta_p$. We claim that the dynamics subsystem satisfies $\C_2$, and prove so using the framework presented in \cite{SharfADHS2020}. Indeed, \cite{SharfADHS2020} shows that the controlled dynamics subsystem satisfies $\C_2$ if $\theta_{2,1},\theta_{0,0}\le 0$, where:

\small
\begin{align*}
   \theta_{0,0} =  \max ~&~  -(p_m(0) - p_f(0) - hv_f(0) - \delta_p)\\ \nonumber
    {\rm s.t.} ~&~ p_m(0) - p_f(0) - hv_f(0) \ge \delta_p\\ \nonumber
    				~&~ p_m(0),p_f(0),v_m(0),v_f(0) \in \R\\
\\
  \theta_{2,1} =  \max ~&~  -(p_m^+ - p_f^+- hv_f^+ - \delta_p)\\ \nonumber
    {\rm s.t.} ~&~  p_m- p_f - hv_f \ge \delta_p\\ \nonumber
    ~&~ p_m^+ - p_m - \Delta t v_m \le \tau v_m + \xi_{\rm up}\\
    ~&~ p_m^+ - p_m - \Delta t v_m \ge -\tau v_m - \xi_{\rm down}\\
    ~&~ v_m^+ - v_m \in [-\eta_{\rm down},\eta_{\rm up}]\\
    ~&~ p_f^+ = p_f + \Delta t v_f,~~v_f^+ = v_f + \Delta t a_f \\ \nonumber
    ~&~ a_f = \frac{p_{m}-p_{f}-hv_{f}}{h\Delta t}+\frac{(\Delta t-\tau)v_{m}}{h\Delta t}-\frac{v_f}{h}-\frac{\lambda}{h\Delta t}\\\nonumber
    ~&~ p_m,p_m^+,v_m,v_m^+,p_f,p_f^+,v_f,v_f^+,a_f \in \R.
\end{align*}
\normalsize
Solving these problems using Yalmip, we find $\theta_{0,0}=\theta_{2,1} = 0$, meaning that the dynamics subsystem indeed satisfies $\C_2$. 

\begin{figure}[b]
    \centering
    \subfigure[Velocity of leader.] {\scalebox{.32}{\includegraphics{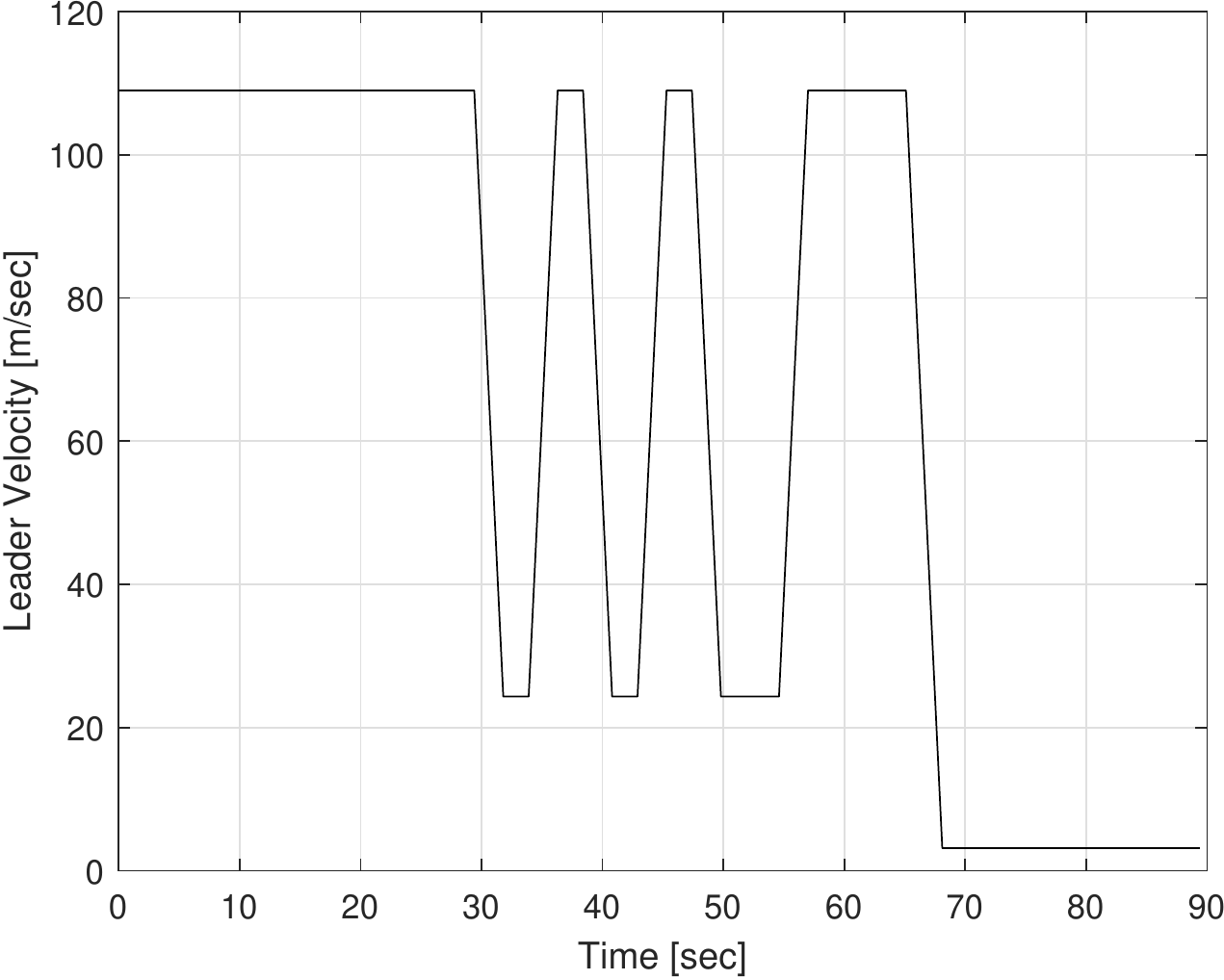}}} \hfill
    \subfigure[Acceleration of leader.] {\scalebox{.32}{\includegraphics{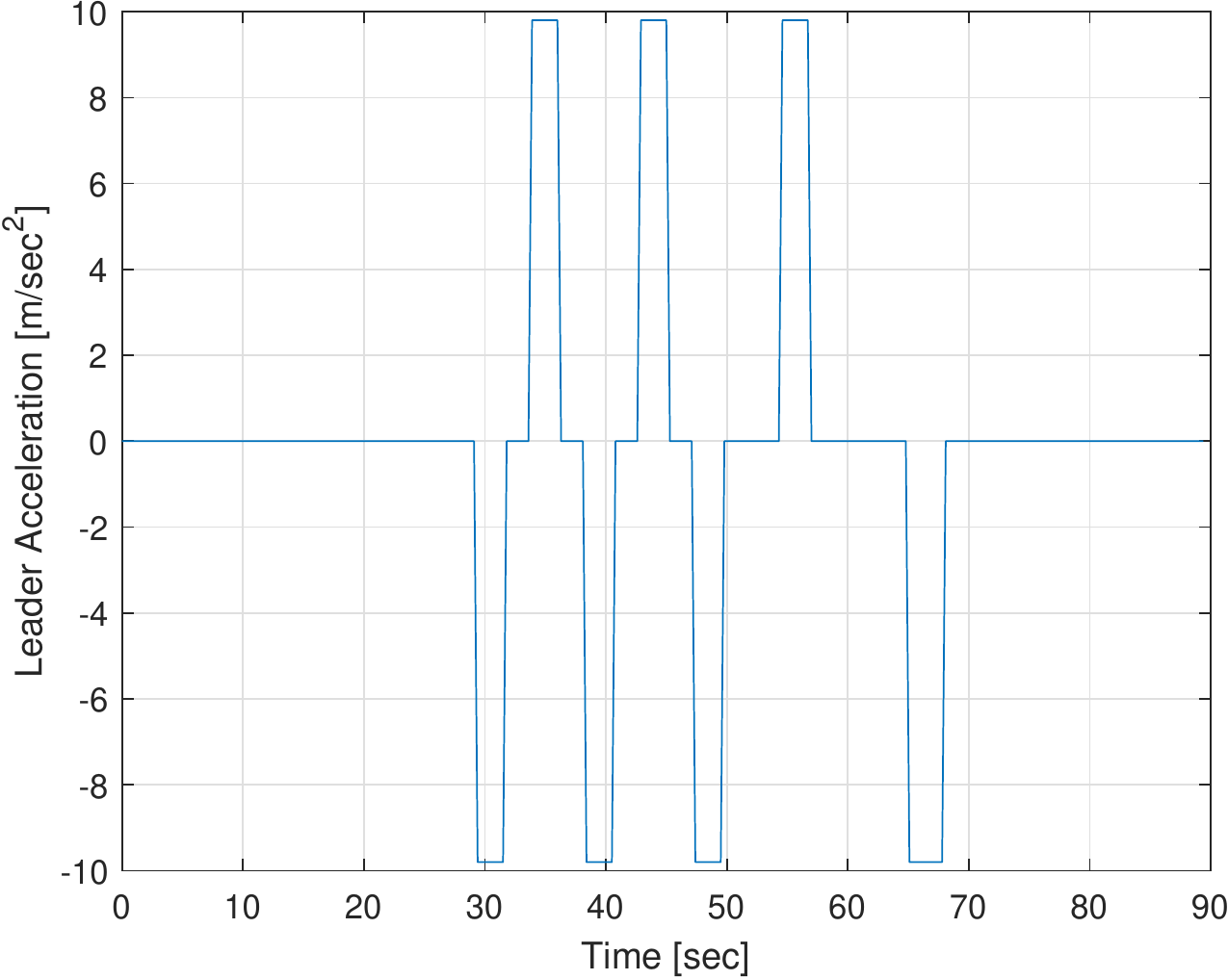}}}
    \caption{Leader vehicle in simulation.}
    \label{fig.Leading}
    \vspace{-15pt}
\end{figure}

\begin{figure}[t]
    \centering
    \includegraphics[width=0.45\textwidth]{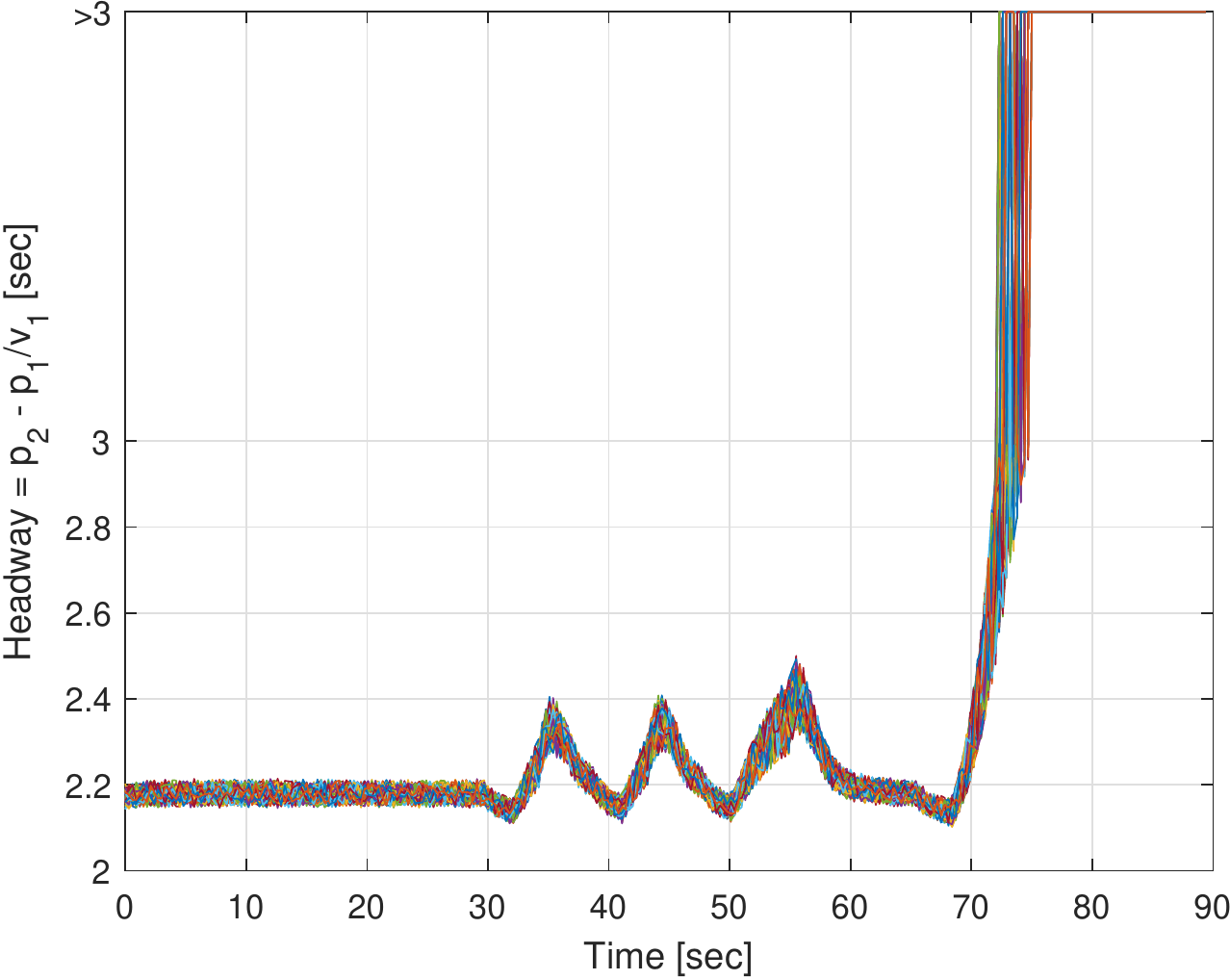}
    \caption{Headway between vehicles.}
    \label{fig.Headway}
\end{figure}

\begin{figure}[b]
    \centering
    \includegraphics[width=0.45\textwidth]{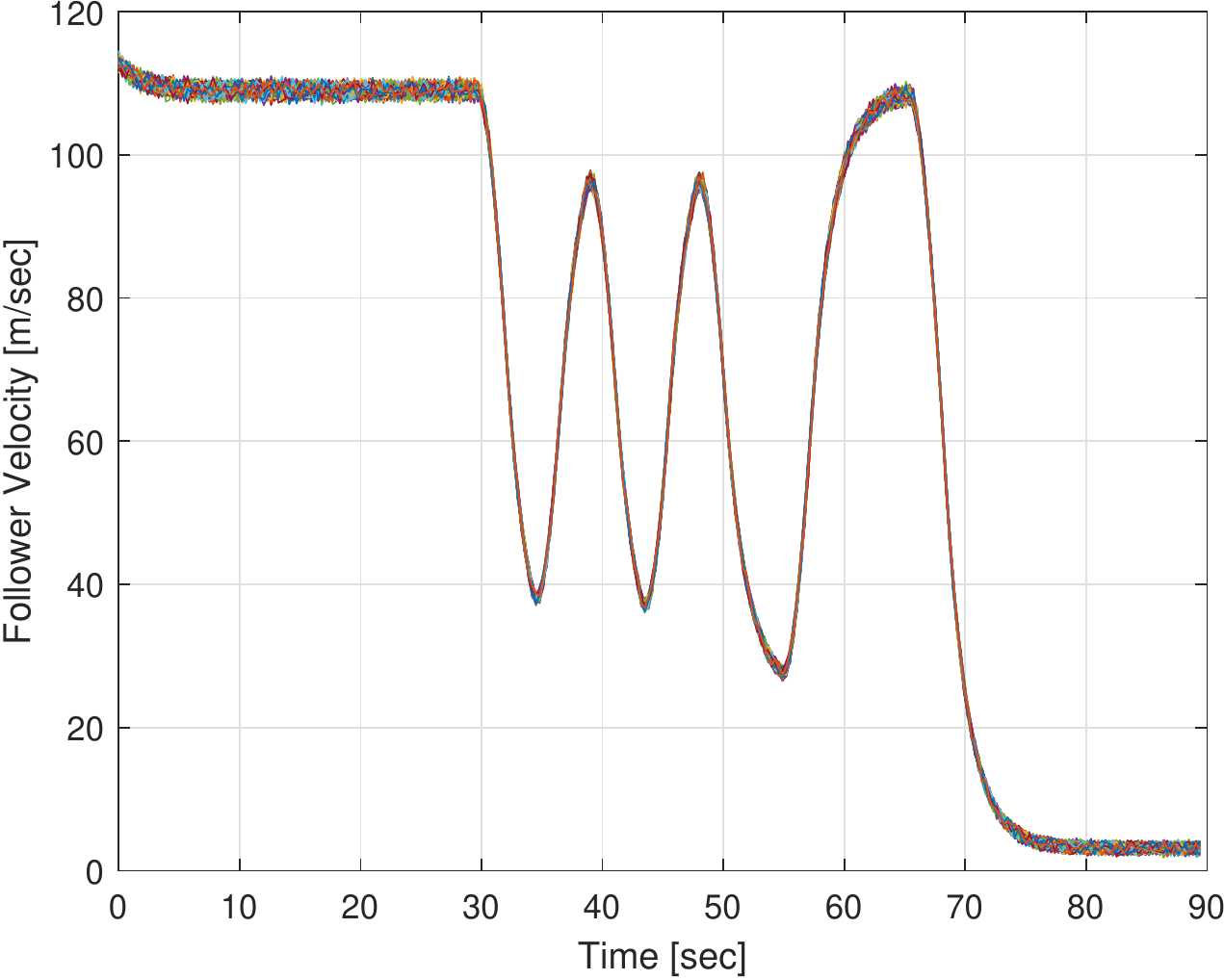}
    \caption{Velocity of follower vehicle.}
    \label{fig.FollowerVelocity}
    \vspace{-15pt}
\end{figure}

We now run a simulation of the closed-loop system, and verify that the safety specifications are met. The simulation takes a total of $90$ seconds. The leader starts at about $110{\rm km/h}$, and roughly keeps this speed for $30$ seconds. Between time $t=30{\rm sec}$ and $t = 60{\rm sec}$, the leader sways wildly between $25{\rm km/h}$ and $110{\rm km/h}$, braking and accelerating as hard as possible. At time $t=65{\rm sec}$, the leader brakes as hard as possible, lowering its velocity to about $3{\rm km/h}$. It then keeps this velocity until the end of the run. The velocity and acceleration of the leader can be seen in Fig. \ref{fig.Leading}. 
The follower starts $70 \rm m$ behind the leader at a velocity of about $113{\rm km/h}$, so the headway is kept at time $0$. We run the simulation $100$ times for both vehicles, where we simulate both the perception and the dynamics for the follower, and choose the random delay and noise independently in each run. The headway $\frac{p_l(k)-p_f(k)}{v_f(k)}$ and the velocity of the follower appear in Fig. \ref{fig.Headway} and Fig. \ref{fig.FollowerVelocity} respectively. It can be seen that in all cases, the headway is kept throughout the run, so the guarantees are satisfied, as predicted by our analysis. 

\section{Conclusions}
Building upon the assume/guarantee framework presented in \cite{SharfADHS2020}, we presented a computationally viable method of verifying compositional refinement of cascaded contracts, assuming the contracts are defined using linear inequalities. Indeed, we showed that compositional refinement is equivalent to three implications, and restated these implications using linear programs. This allows one to check whether a pair of cascaded contracts refines a contract on the composite system by solving a few linear programs. The amount of linear programs needed to be solved grows linearly with the number of assumptions and guarantees of the contracts in the problem. Similar methodology can be used to verify compositional refinement for a larger cascade of systems using the contract algebra defined in \cite{SharfADHS2020}. We exemplified our findings with a thorough case study of a 2-vehicle leader-follower scenario, in which we showed that safety specifications are met even under noise and time-varying delay. Future research can extend these results by building similar tools for verification of compositional refinement for feedback compositions or contracts for hybrid systems.

\bibliographystyle{ieeetran}
\bibliography{main}
\end{document}